\pdfoutput=1
\documentclass[10pt,conference]{IEEEtran}
\IEEEoverridecommandlockouts
\usepackage{cite}
\usepackage{amsmath,amssymb,amsfonts,amsthm}
\usepackage{algorithmic}
\newtheorem{theorem}{Theorem}

\newtheorem{definition}{Definition}
\usepackage{graphicx}
\usepackage{textcomp}
\usepackage{xcolor}
\usepackage{booktabs}
\usepackage{multirow}
\usepackage{url}
\usepackage{hyperref}
\usepackage{caption}
\usepackage{eso-pic}
\hypersetup{hidelinks}
\def\BibTeX{{\rm B\kern-.05em{\sc i\kern-.025em b}\kern-.08em
    T\kern-.1667em\lower.7ex\hbox{E}\kern-.125emX}}
\begin{document}

\title{Measurement-Device Placement in MDI-QKD Networks with Minimal Trusted Relays}

\author{
\IEEEauthorblockN{Tianqu Luo}
\IEEEauthorblockA{
Politecnico di Milano, Italy \\
tianqu.luo@mail.polimi.it
}
\and
\IEEEauthorblockN{Ibrahim Arslantas}
\IEEEauthorblockA{
Politecnico di Milano, Italy \\
ibrahim.arslantas@mail.polimi.it
}
\and
\IEEEauthorblockN{Sezin Ozturk}
\IEEEauthorblockA{
Politecnico di Milano, Italy \\
sezin.ozturk@mail.polimi.it
} 
\and
\IEEEauthorblockN{Qiaolun Zhang}
\IEEEauthorblockA{
Politecnico di Milano, Italy \\
qiaolun.zhang@mail.polimi.it \\
\textit{Corresponding author}
}
\and
\IEEEauthorblockN{Alberto Sebastián-Lombraña}
\IEEEauthorblockA{
Universidad Politécnica de Madrid, Spain \\
Center for Computational Simulation, Madrid, Spain\\
aj.sebastian@upm.es
}
\and
\IEEEauthorblockN{Mehdi Bolourian}
\IEEEauthorblockA{
University of Waterloo, Canada \\
mbolourian@uwaterloo.ca
}
\and
\IEEEauthorblockN{Jiaheng Xiong}
\IEEEauthorblockA{
Politecnico di Milano, Italy \\
jiaheng.xiong@polimi.it
}
\and
\IEEEauthorblockN{Francesco Musumeci}
\IEEEauthorblockA{
Politecnico di Milano, Italy \\
francesco.musumeci@polimi.it
}
\and
\IEEEauthorblockN{Vicente Martín}
\IEEEauthorblockA{
Universidad Politécnica de Madrid, Spain \\
Center for Computational Simulation, Madrid, Spain\\
vicente.martin@upm.es
}
\and
\IEEEauthorblockN{Raouf Boutaba}
\IEEEauthorblockA{
University of Waterloo, Canada \\
rboutaba@uwaterloo.ca
}
\and
\IEEEauthorblockN{Massimo Tornatore}
\IEEEauthorblockA{
Politecnico di Milano, Italy \\
massimo.tornatore@polimi.it
}
}
\maketitle
\AddToShipoutPictureFG*{\AtPageLowerLeft{\put(0,\LenToUnit{0.42in}){\makebox[\paperwidth]{\parbox{\textwidth}{\centering \copyright~Owner/Author, IEEE. This is the author's version of the work. It is posted here for your personal use. Not for redistribution. The definitive Version of Record was accepted by IEEE QCE 2026.}}}}}

\begin{abstract}

 Measurement-Device-Independent Quantum Key Distribution (MDI-QKD) removes detector side-channel vulnerabilities by delegating measurements to an untrusted relay that, when shared by several user pairs, acts as a Bell-state measurement (BSM) hub. Channel loss still limits the reach of MDI-QKD links, so long-range services rely on trusted relays for key forwarding. As MDI-QKD moves toward metropolitan deployment, a key network-planning question arises: how should BSM hubs be placed on existing fiber infrastructure to minimize the use of trusted relays?
 In this work, we formalize this challenge as the MDI-QKD Hub Deployment (MHD) problem. We first prove that the MHD problem is NP-hard, and then formulate it as a Mixed-Integer Linear Programming (MILP). To the best of our knowledge, this is the first formulation for MDI-QKD network planning that captures the structural features specific to MDI-QKD, namely a shared BSM hub, two-link user-to-hub routing, and loss-balancing constraints, while jointly determining hub placement, user-to-hub assignment, and trusted-relay demand allocation. Our solution accounts for realistic budget and capacity constraints, including practical insights from real-world deployments of commercial MDI-QKD solutions. Numerical evaluations on three metropolitan topologies (12 to 50 nodes) at realistic geographic distances show that topology structure governs trusted-relay demand: hub placement alone eliminates all trusted relays in compact urban meshes, while in sparse topologies, relaxing the loss-balancing constraint removes the dominant source of trusted-relay usage at no additional infrastructure cost. The key-rate threshold at which trusted relays first appear shifts monotonically with the network diameter, i.e., the largest shortest-path distance between any node pair in fiber kilometers, delineating the feasibility boundary of pure MDI-QKD deployment.
\end{abstract}


\section{Introduction}
\label{sec:intro}

Quantum key distribution (QKD) can provide information-theoretically secure key exchange~\cite{Bennett1984, Lo2014, Pirandola2017}, but its transmission range is limited by channel losses~\cite{Pirandola2017}. To extend QKD services over metropolitan and backbone networks, current deployments rely on trusted relay nodes that forward keys hop by hop~\cite{qkd_a_networking}, at the cost of introducing additional security assumptions at each relay.

The need to extend QKD services beyond point-to-point links has motivated substantial effort toward integrating QKD into optical fiber infrastructures, including studies on resource allocation, routing, and key-rate optimization in networked settings~\cite{Zhao2018, Zhang2024RCKTA, xiong2025power}. Over the past two decades, QKD technology has progressed from laboratory experiments to deployed metropolitan and backbone networks, confirming its practical viability for securing critical communications. Currently, several efforts for the deployment of quantum network infrastructures to support QKD-enabled secure communications are ongoing, such as those related to the EuroQCI initiative~\cite{web_euroqci, euroqci_QUID_Italia, euroqci_spain}.

Nevertheless, the information-theoretic security of QKD assumes that detectors behave as ideal single-photon counters. Real-world detectors deviate from this model through imperfections such as efficiency mismatch, dead time, and afterpulsing~\cite{Lo2012}, and they remain vulnerable to active attacks such as detector blinding~\cite{Lydersen2010}, all of which open side channels that can completely compromise security. Measurement-Device-Independent QKD (MDI-QKD)~\cite{Lo2012, Braunstein2012} overcomes this vulnerability by delegating all measurements to an untrusted relay performing Bell-state measurements (BSMs), thereby removing detector side channels. Since its introduction, MDI-QKD has been experimentally demonstrated over long fiber distances~\cite{Tang2014} and in metropolitan-scale optical networks~\cite{Tang2016}, including coexistence with classical data traffic~\cite{Berrevoets2022}. Some commercial MDI-QKD systems, notably those by Q*Bird~\cite{Berrevoets2022, qbird_web}, employ a shared measurement relay that can serve multiple user pairs, effectively functioning as a central hub.

These advances naturally raise a network design question: 

\emph{how should BSM hubs be placed on existing fiber infrastructure to minimize the use of trusted relays while delivering long-range QKD services?}
This can be viewed as a hub location problem~\cite{Alumur2008, OKelly1987}, where each hub denotes a shared BSM relay serving multiple user pairs, but with several constraints specific to MDI-QKD. First, communication between two users is realized via a \emph{two-link connection to a shared relay}, i.e., each user independently establishes a quantum channel to the same BSM hub. Second, the achievable secret key rate decays rapidly with channel loss, which effectively limits the maximum transmission distance. Third, practical MDI-QKD performance is sensitive to imbalance between the two paths from the user pair to the BSM hub. For a given user pair (conventionally referred to as Alice and Bob), significant mismatch between the two optical paths to the same BSM hub degrades two-photon interference (i.e., the Hong-Ou-Mandel interference~\cite{hong1987measurement} that enables BSM) and reduces the achievable key rate. A common approach is therefore to compensate this asymmetry by adding attenuation to the lower-loss path, effectively enforcing balanced end-to-end losses. 

In this work, we adopt this symmetric approach~\cite{Tang2016, Berrevoets2022} and treat loss equalization as a deployment constraint: a user pair can be assigned to a BSM hub only if the two paths from the users to that hub can be loss-balanced, by attenuating the lower-loss path, without pushing the effective distance beyond the key-rate threshold. Because balancing always raises both channels to the higher-loss level, hub placements with highly asymmetric user-to-hub paths may become infeasible even when each individual channel is within range. This assumption captures standard experimental practice and enables a tractable optimization formulation, though it does not cover asymmetric MDI-QKD protocols~\cite{Wang2019, Liu2019}.

As interest in MDI-QKD network deployment grows, a network planning model that captures realistic deployment constraints becomes essential to optimize both infrastructure placement and ongoing service provision. Prior studies on MDI-QKD deployment rely on simplified system models: Cao et al.~\cite{Cao2021} select relay types along predetermined paths without jointly modeling hub placement and routing, while Jia et al.~\cite{Jia2023} and He et al.~\cite{He2025} employ heuristic frameworks that do not model the two-link user-to-hub routing structure or realistic budget and capacity constraints. On the other hand, ILP and MILP formulations for adjacent quantum-network problems, such as repeater placement~\cite{Rabbie2022} and resource placement~\cite{Pouryousef2024}, do not capture the constraints specific to MDI-QKD, namely a two-link connection to a shared BSM hub, distance-dependent key-rate decay, and loss-balancing constraints on the two user-to-hub paths. As a result, no existing framework jointly addresses hub placement, two-link user-to-hub routing, and trusted-relay assignment under realistic deployment constraints.

The contributions of this study can be summarized as follows.
\begin{itemize}
    \item \textbf{Problem formulation and complexity.} We introduce the MDI-QKD Hub Deployment (MHD) problem, which jointly determines the placement of BSM hubs in the fiber-based QKD network, the assignment of user pairs to those hubs, and the use of trusted-relay paths when direct MDI-QKD is infeasible, so as to satisfy the expected service demand under operational conditions. We show that the MHD problem is NP-hard.    
    
    \item \textbf{MILP-based optimal design.} We formulate the MHD problem as a Mixed-Integer Linear Programming (MILP) model that explicitly captures the two-link connection to a shared relay of MDI-QKD connections under loss-balancing constraints. Channel asymmetry is handled via a loss-balancing constraint based on passive attenuation, yielding a tractable yet physically grounded approximation. The resulting formulation can be solved to optimality for metropolitan-scale deployments.
    
    \item \textbf{Network-level insights.} Through simulations on three metropolitan topologies (12 to 50 nodes) at realistic geographic distances under three candidate placement strategies of increasing density, we show that topology structure governs trusted-relay demand: compact meshes eliminate all trusted relays through hub placement alone, sparse topologies benefit substantially from relaxing the loss-balancing constraint, and the key-rate threshold at which trusted relays first appear shifts monotonically with network diameter, delineating explicit feasibility boundaries for deployment planning.
\end{itemize}

\section{Related Work}
\label{sec:related}

Metropolitan and backbone QKD networks have been deployed in multiple continents, from the DARPA testbed in Boston~\cite{boston}, the SECOQC network in Vienna~\cite{Peev2009} and the Tokyo QKD network~\cite{Sasaki2011} to the 2,000\,km Beijing--Shanghai backbone~\cite{Chen2021}. Although these deployments validated QKD in real-world settings, they relied on trusted intermediate nodes, each of which must be assumed secure since a compromise exposes the keys it relays. MDI-QKD eliminates detector side-channel vulnerabilities by delegating measurements to an untrusted relay; recent deployments include a metropolitan MDI-QKD network~\cite{Tang2016} and the Madrid quantum network~\cite{madrid_2025_advancing}, in both cases using a hub-based approach.

A growing body of research is now focusing on how to efficiently allocate resources in QKD networks through mathematical optimization. ILP and MILP formulations have been proposed for quantum-repeater placement~\cite{Rabbie2022}, resource placement~\cite{Pouryousef2024}, and joint routing and key-rate assignment~\cite{Zhang2024RCKTA}, while other studies have investigated key management, resource allocation, and finite-key security analyses~\cite{Zhang2023ICQKD, Cirigliano2024, Amer2020, Krawec2024, Marik2025}. These formulations target relay-chain architectures, whether trusted-node key forwarding or quantum-repeater entanglement distribution, and thus do not apply to MDI-QKD, where each user pair connects to a shared BSM hub via two independent links.

The few prior studies that target MDI-QKD-specific deployment~\cite{Cao2021, Jia2023, He2025} address only restricted subproblems and do not jointly model hub placement, two-link user-to-hub routing, and trusted-relay assignment under realistic constraints.

\section{System Model and Problem Statement}
\label{sec:model}

\subsection{Network Model}
\label{sec:network_model}

We model a metropolitan fiber-based QKD network as a directed graph $G = (V, A)$, where $V$ is the set of nodes and $A$ is the set of directed arcs representing fiber links, each with physical length $\ell_a$ in kilometers. Each physical fiber link is represented by a pair of opposite directed arcs, hence the arc-set notation $A$ rather than an undirected edge set. A set $\mathcal{C}$ denotes \emph{candidate hub locations} where Bell-state measurement (BSM) equipment may be installed.

\begin{figure*}[t]
\centering
\includegraphics[width=0.9\textwidth]{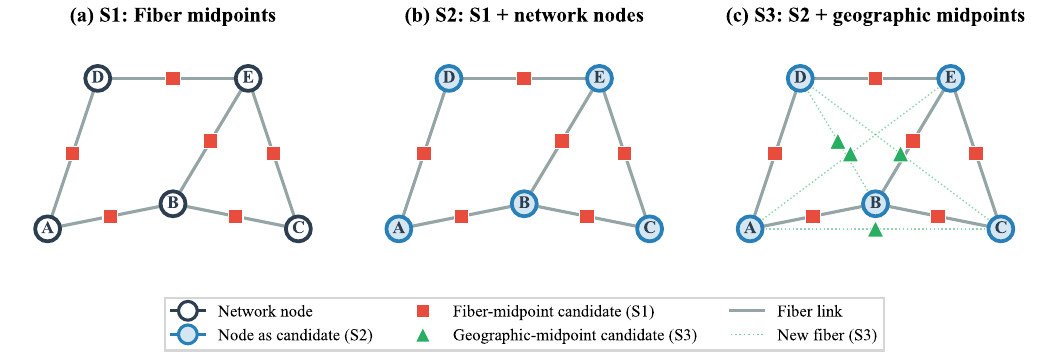}
\caption{Candidate hub locations under three placement strategies
on a five-node example topology.}
\label{fig:candidate_strategies}
\end{figure*}

The set $\mathcal{C}$ is determined by a \emph{candidate selection policy} reflecting deployment constraints:

\begin{itemize}
    \item \textbf{S1 (Fiber Midpoints):} Candidates are placed at midpoints of existing fiber links (Fig.~\ref{fig:candidate_strategies}(a)), reusing splice points or amplifier sites without new facility construction. Under homogeneous fiber, such locations tend to reduce channel-loss asymmetry.

    \item \textbf{S2 (Fiber Midpoints + Network Nodes):} Candidates include all S1 locations plus all network nodes $V$ (Fig.~\ref{fig:candidate_strategies}(b)), additionally leveraging existing switching facilities for co-locating BSM equipment.

    \item \textbf{S3 (S2 + Geographic Midpoints):} Candidates include all S2 locations plus the geographic midpoint of every non-adjacent node pair (up to $\binom{|V|}{2}$ additional locations), subject to a minimum separation of 2\,km (Fig.~\ref{fig:candidate_strategies}(c)); reaching these sites requires new fiber whose cost is not captured by~$B$. 
    Since the candidate hub locations $\mathcal{C}$ under S3 include all the candidate hub locations under S1 and S2, 
    S3 provides a lower bound on achievable trusted-relay usage and helps distinguish \emph{topology-limited} requests from \emph{distance-limited} ones; it should therefore be read as a theoretical reference rather than a deployment option on an equal cost footing with S1 and S2.
\end{itemize}

A set of \emph{key distribution requests} $\mathcal{R} = \{(s_r, d_r)\}_{r=1}^{|\mathcal{R}|}$ specifies source--destination pairs requiring quantum-distributed keys. Each request must be served either via an MDI-QKD connection through a BSM hub or, when such service is not feasible, via trusted-relay forwarding through existing network nodes under different security assumptions. These two options differ fundamentally in their trust model. A BSM hub is \emph{untrusted}: it only performs the Bell-state measurement and publicly announces the outcome, learning nothing about the secret key, so compromising a hub does not expose key material. A trusted relay, in contrast, receives, stores, and re-forwards key bits and must therefore be \emph{trusted}, since its compromise reveals every key it relays. We treat BSM hubs as intermediate measurement stations distinct from the communicating endpoints: a network node may host a BSM hub, but a user pair is never served by a hub co-located with either of its own endpoints. Since trusted-relay operation introduces additional trust assumptions at the relay nodes, the primary optimization objective is to minimize the number of requests served with trusted relays (TR).
Deploying BSM equipment at a candidate location~$h$ incurs a fixed installation cost~$c_h$, and routing a user pair through an active hub incurs an additional per-connection cost~$c_u$. The aggregate deployment expenditure is bounded by a budget~$B$. Each BSM unit can serve at most $K_\text{max}$ concurrent requests; deploying multiple units at the same site increases its capacity proportionally.

\subsection{Key-Rate Model and Deployment Assumptions}
\label{sec:keyrate}

In an MDI-QKD link, each of the two users (defined as Alice and Bob) acts as a \emph{transmitter}, independently sending prepared quantum states through a dedicated fiber channel to a shared BSM hub whose detection apparatus serves as the \emph{receiver}. The secret key rate is governed by the optical losses in these two transmitter-to-receiver channels. Under the decoy-state MDI-QKD protocol~\cite{Lo2012}, the asymptotic key rate per pulse is
\begin{equation}
    R \geq Q_{11}^Z \bigl[1 - H(e_{11}^X)\bigr] - Q_{\mu\mu}^Z\, f_e\, H(E_{\mu\mu}^Z),
    \label{eq:keyrate_physics}
\end{equation}
where $Q_{11}^Z$ and $e_{11}^X$ denote the single-photon gain and phase-error rate estimated via decoy-state analysis, $Q_{\mu\mu}^Z$ and $E_{\mu\mu}^Z$ are the overall gain and QBER, $H(\cdot)$ is the binary entropy function, and $f_e \approx 1.16$ is the error-correction inefficiency.

In standard fibers with attenuation $\alpha_\text{dB} \approx 0.2$\,dB/km, channel transmittance decays exponentially with distance. In MDI-QKD, the secret-key rate is highly sensitive to imbalance between the two user-to-BSM channels, because large asymmetry degrades the effective two-photon interference conditions at the relay and thereby limits performance. Consequently, BSM-hub placement is a critical design variable, as it directly determines the degree of channel-loss asymmetry.

For tractability, rather than evaluating~\eqref{eq:keyrate_physics}
for each candidate configuration, we
follow~\cite{Berrevoets2022} and model the key rate $\kappa$ as a
function of an effective distance $d_\text{eff}$:
\begin{equation}
    \kappa = f(d_\text{eff}), \quad \text{viable iff } \kappa \geq \kappa_\text{min},
    \label{eq:keyrate}
\end{equation}
where $f(\cdot)$ is implemented via interpolation over experimentally reported data points (50--250\,km, 5\,km spacing), links beyond the maximum tabulated reach $d_\text{max}$ are treated as unreachable, and $\kappa_\text{min}$ is a minimum QoS threshold. Representative values are shown in Table~\ref{tab:keyrate}. The exponential decay reflects fiber attenuation at 1550\,nm.

\begin{table}[t]
\centering
\caption{MDI-QKD key rate as a function of total channel loss, expressed as effective distance $d_\text{eff}$~\cite{Berrevoets2022, qbird_web}.}
\label{tab:keyrate}
\begin{tabular}{lcccccc}
\toprule
$d_\text{eff}$ (km) & 50 & 100 & 150 & 200 & 225 & 250 \\
\midrule
$\kappa$ (bps) & 1300 & 110 & 8.5 & 0.71 & 0.20 & 0.02 \\
\bottomrule
\end{tabular}
\end{table}

Throughout this work, path asymmetry is modeled only through the dominant loss mechanisms in fiber-based MDI-QKD networks: fiber attenuation and node-bypass loss. Connector losses, splice losses, and other intermediate-node component losses are neglected; $\alpha_\text{dB}$ could be easily increased to include the cumulative effect of these losses. For user $i \in \{A,B\}$, let $d_i$ denote the fiber length from the user to the BSM hub, and let $N_i$ denote the number of intermediate nodes traversed along that path. The corresponding channel loss is modeled as $L_i = d_i \, \alpha_\text{dB} + N_i \, b_l$ (in dB), where $\alpha_\text{dB}$ is the fiber attenuation coefficient and $b_l$ is the bypass loss incurred at each intermediate node.
Based on these channel losses, we define the effective distance as
\begin{equation}
    d_\text{eff} = \frac{2}{\alpha_\text{dB}} \max(L_A, L_B),
    \label{eq:deff_bypass}
\end{equation}
which maps the larger of the two user-to-BSM channel losses to an equivalent end-to-end fiber distance, with the division by $\alpha_\text{dB}$ (in dB/km) converting the loss in dB back into kilometers. Under loss-balanced operation, this expression reduces to the total equivalent distance across the two user-to-BSM paths.

\noindent\textbf{Deployment options.}
We consider two deployment options that differ in how channel-loss asymmetry is handled.

\textit{Loss-compensated deployment.}
This deployment option assumes that channel-loss asymmetry can be actively mitigated (e.g., via optical attenuation or protocol-level compensation), so that the two-link connections from the user pair to the BSM hub operate under effectively balanced loss. We enforce $|L_A - L_B| \leq 2\tau \cdot \alpha_\text{dB}$, which, under the homogeneous fiber assumption, is equivalent to $|d_A - d_B| \leq 2\tau$. We set $\tau = 3$\,km, corresponding to at most $1.2$\,dB of residual loss imbalance. This is consistent with typical optical penalty budgets and is conservative relative to prior theoretical analyses~\cite{Avis2024} and experimental demonstrations~\cite{Tang2016} showing robust MDI-QKD operation under significantly larger asymmetries. A sensitivity sweep over $\tau$ values between 1 and 10\,km across the evaluated instances shows the expected monotone behavior, i.e., tightening the window increases trusted-relay usage while widening it gradually relaxes the deployment toward the loss-uncompensated model, with the effect concentrated in sparse topologies.

\textit{Loss-uncompensated deployment.}
In this model, no loss-balancing is enforced: any BSM hub reachable within the key-rate threshold is admissible, regardless of the imbalance between the two user-to-hub paths. We retain the same key-rate model and evaluate the rate conservatively at the worse of the two channels through the effective distance $d_\text{eff} = (2/\alpha_\text{dB})\max(L_A, L_B)$ of~\eqref{eq:deff_bypass}, which over-estimates the effective loss whenever the two paths differ. The resulting key rate is therefore a lower bound and the trusted-relay counts are conservative. The two deployment options thus differ only in candidate admissibility; neither models an asymmetric MDI-QKD protocol.

\subsection{Problem Statement}
\label{sec:problem}

\begin{definition}[MDI-QKD Hub Deployment (MHD) Problem]
\label{def:mhdp}
The MHD problem can be formally stated as follows: given a fiber-based QKD network $G=(V,A)$, candidate hub locations $\mathcal{C}$, key distribution requests $\mathcal{R}$, deployment budget $B$, per-hub capacity $K_\text{max}$, and a key-rate function $f(\cdot)$ with minimum threshold $\kappa_\text{min}$, determine a hub deployment $\mathcal{H} \subseteq \mathcal{C}$ (with multiplicities) and a routing assignment $\phi: \mathcal{R} \to \bigcup_r \mathcal{P}_r$, where $\mathcal{P}_r$ is the set of candidate serving paths for request~$r$ (defined formally in Section~\ref{sec:variables}), that minimizes the number of requests served via trusted relays, i.e., $|\{r \in \mathcal{R} : \phi(r) \text{ is a trusted-relay path}\}|$, subject to budget, hub capacity, arc capacity, and key-rate constraints. The MHD problem simultaneously decides \emph{where} to deploy BSM equipment, \emph{how} to route the quantum path across each user pair, and \emph{which} pairs must use trusted relays. This joint optimization distinguishes it from sequential approaches that first fix hub locations and then optimize routing~\cite{Jia2023, He2025}. We address the \emph{static planning} phase, in which hub placement and the user-to-hub routing it must support are decided jointly and offline, prior to deployment; dynamic, per-session key routing over an already-provisioned infrastructure operates at a different time scale and is outside the scope of this work.
\end{definition}

\section{MILP Formulation}
\label{sec:milp}

\subsection{Decision Variables and Parameters}
\label{sec:variables}

Following the model in Section~\ref{sec:model}, the candidate path set for each request~$r$ is
\begin{equation}
    \mathcal{P}_r = \{p_{r,h} \mid h \in \mathcal{C},\; \kappa(d_\text{eff}^{r,h}) \geq \kappa_\text{min}\} \cup \{p_r^\text{TR}\},
    \label{eq:pathset}
\end{equation}
where $p_{r,h}$ denotes the two-link MDI-QKD path through hub~$h$ and $p_r^\text{TR}$ the path that uses trusted relay.

We define the following decision variables:
\begin{itemize}
    \item $Y_h \in \mathbb{Z}_{\geq 0}$: number of BSM units deployed at candidate location $h \in \mathcal{C}$. Multiple units at a single site increase its concurrent request-serving capacity, modeling scenarios where a facility hosts several BSM receivers.
    \item $X_{r,p} \in \{0,1\}$: equals~1 if request $r \in \mathcal{R}$ is assigned to path $p \in \mathcal{P}_r$, and 0 otherwise.
\end{itemize}

The model uses the following input parameters:
\begin{itemize}
    \item $\delta^\text{Hub}_{p,h} \in \{0,1\}$: equals~1 if path~$p$ uses hub~$h$ for BSM.
    \item $\delta^\text{TR}_{p} \in \{0,1\}$: equals~1 if path~$p$ is a trusted-relay path.
    \item $c_h$: fixed cost for deploying BSM equipment at hub~$h$.
    \item $c_u$: per-connection cost for each user pair routed through a hub.
    \item $n_p^\text{ch}$: number of quantum channels consumed by path~$p$.
    \item $K_\text{max}$: maximum number of requests a single BSM unit can serve.
    \item $B$: upper bound on aggregate hub deployment and connection cost.
\end{itemize}

\subsection{Objective Function}
\label{sec:objective}

MDI-QKD network planning involves three competing objectives: minimizing the number of requests that are served with trusted relays (security), minimizing deployment cost (economics), and minimizing quantum channel usage (resource efficiency). We combine these into a single weighted objective with lexicographic-style priorities:
\begin{align}
    \min \quad & \alpha \!\sum_{r \in \mathcal{R}} \sum_{p \in \mathcal{P}_r} \!\delta^\text{TR}_{p} X_{r,p} \notag \\
    & + \beta \!\left(\sum_{h \in \mathcal{C}} c_h Y_h + c_u \!\sum_{r,p} n_p^\text{hub} X_{r,p}\right) \notag \\
    & + \gamma \sum_{r \in \mathcal{R}} \sum_{p \in \mathcal{P}_r} n_p^\text{ch} X_{r,p}, \label{eq:obj}
\end{align}
where $n_p^\text{hub} = \sum_{h \in \mathcal{C}} \delta^\text{Hub}_{p,h}$ is the number of BSM hubs used by path~$p$.

The weight hierarchy $\alpha = 1000 \gg \beta = 1 \gg \gamma = 0.001$ enforces strict priority: security (TR minimization) is non-negotiable, cost is secondary, and channel count is a tiebreaker. This choice makes the weighted objective equivalent to a strict lexicographic order (TR~$\succ$~cost~$\succ$~channels) refining the trusted-relay minimization objective of Definition~\ref{def:mhdp}: the cost term is bounded by $\beta B = B < \alpha$ and the channel term satisfies $\gamma \sum_{r,p} n_p^\text{ch} X_{r,p} < 1$ across all evaluated instances, so the second and third terms together can never offset a single-unit change in the trusted-relay count.

\subsection{Constraints}
\label{sec:constraints}

\textbf{Connectivity.} Every request is served by exactly one path:
\begin{equation}
    \sum_{p \in \mathcal{P}_r} X_{r,p} = 1, \quad \forall\, r \in \mathcal{R}.
    \label{eq:serve}
\end{equation}

\textbf{Hub activation.} A path using hub $h$ requires deployed BSM equipment:
\begin{equation}
    X_{r,p} \leq Y_h, \quad \forall\, r,\, p,\, h \;:\; \delta^\text{Hub}_{p,h}\!=\!1.
    \label{eq:activate}
\end{equation}
This per-path formulation is tighter than big-$M$ alternatives~\cite{Alumur2008}, producing a stronger LP relaxation.

\textbf{Hub capacity.} Each BSM unit serves at most $K_\text{max}$ requests:
\begin{equation}
    \sum_{r \in \mathcal{R}} \sum_{\substack{p \in \mathcal{P}_r \\ \delta^\text{Hub}_{p,h}=1}} \!\! X_{r,p} \leq K_\text{max} \cdot Y_h, \quad \forall\, h \in \mathcal{C}.
    \label{eq:capacity}
\end{equation}

\textbf{Budget limit.}
\begin{equation}
    \sum_{h \in \mathcal{C}} c_h Y_h + c_u \sum_{r,p} n_p^\text{hub} X_{r,p} \leq B.
    \label{eq:budget}
\end{equation}

\textbf{Arc capacity.} Each directed fiber arc carries at most $W$ quantum channels:
\begin{equation}
    \sum_{r \in \mathcal{R}} \sum_{\substack{p \in \mathcal{P}_r \\ \delta^\text{Arc}_{p,a}=1}} X_{r,p} \leq W, \quad \forall\, a \in A,
    \label{eq:arccap}
\end{equation}
where $\delta^\text{Arc}_{p,a} = 1$ if path~$p$ uses arc~$a$. This prevents over-subscription of quantum channels on shared fiber links.

\textbf{Key-rate prefiltering.} Paths with $\kappa(d_\text{eff}^p) < \kappa_\text{min}$ are excluded during path generation, keeping the MILP constraint matrix sparse.

\subsection{NP-hardness of the MHD Problem}
\label{sec:complexity}

We establish the computational hardness of the MHD problem by reduction from the Set Cover problem, a classical NP-hard problem~\cite{Karp1972}.

\begin{theorem}[Computational Hardness]
\label{thm:hard}
The decision version of the MHD problem is NP-complete: given a deployment budget~$B$, it is NP-complete to determine whether there exists a feasible hub deployment under which all requests are served via MDI-QKD paths, i.e., no request needs a trusted relay. 

\end{theorem}

\begin{proof}
\textit{The MHD problem is in NP.} Given a candidate hub deployment and routing assignment, feasibility can be verified in polynomial time by checking path assignment, hub activation, capacity, budget constraints, and counting the number of trusted-relay requests.

\textit{NP-hardness.} We reduce from the Set Cover problem~\cite{Karp1972}, defined as follows: given a universe $U = \{u_1, \ldots, u_m\}$, a collection of subsets $\mathcal{S} = \{S_1, \ldots, S_n\} \subseteq 2^U$, and an integer~$k$, determine whether there exists a sub-collection $I \subseteq \mathcal{S}$ with $|I| \leq k$ such that $\bigcup_{S \in I} S = U$.

Given a Set Cover instance $(U, \mathcal{S}, k)$, we construct an MHD problem instance as follows. Each element $u_i \in U$ is mapped to a request $r_i \in \mathcal{R}$. Each subset $S_j \in \mathcal{S}$ is mapped to a candidate hub $h_j \in \mathcal{C}$. Using a standard distance-threshold construction, the network topology and key-rate parameters can be constructed in polynomial time so that request $r_i$ is reachable via hub $h_j$ if and only if $u_i \in S_j$ (e.g., place $h_j$ at distance $L_1$ from both endpoints of $r_i$ if $u_i \in S_j$ and at distance $L_2 > L_1$ otherwise, with $\kappa_\text{min}$ set between $\kappa(2L_2)$ and $\kappa(2L_1)$). Hub costs are set to unity ($c_h = 1$) with zero per-connection cost ($c_u = 0$), the deployment budget is $B = k$, and the hub capacity $K_\text{max} = m$ and arc capacity $W = m$ are chosen large enough to be non-binding.

Suppose the Set Cover instance has a feasible solution $I = \{S_{j_1}, \ldots, S_{j_{k'}}\}$ with $k' \leq k$ covering all elements of~$U$. Deploying hubs at $\{h_{j_1}, \ldots, h_{j_{k'}}\}$ routes every request~$r_i$ through some hub~$h_j$ with $u_i \in S_j$, yielding zero trusted relays within budget~$B = k$. Conversely, suppose the MHD problem instance admits a zero-TR deployment within budget~$B = k$. Each request $r_i$ is served by some hub~$h_j$, which implies $u_i \in S_j$. The set of activated hubs thus corresponds to a sub-collection of $\mathcal{S}$ of size at most~$k$ that covers all of~$U$.

\end{proof}

\section{Simulation Setup}
\label{sec:setup}

We implement the MILP model in AMPL and solve it using IBM ILOG CPLEX~22.1. We evaluate the model across three metropolitan-scale fiber topologies, three candidate placement strategies, and both deployment models (loss-compensated and loss-uncompensated, as defined in Section~\ref{sec:keyrate}), varying the request load and the minimum key-rate requirement, both individually and jointly.

\subsection{Topologies}
\label{sec:topologies}

We evaluate our MILP on three metropolitan-scale network topologies, summarized in Table~\ref{tab:topologies} and visualized in Fig.~\ref{fig:topologies}, that span the range from dense urban core to large regional infrastructure.

\begin{figure*}[t]
    \centering
    \includegraphics[width=0.93\textwidth]{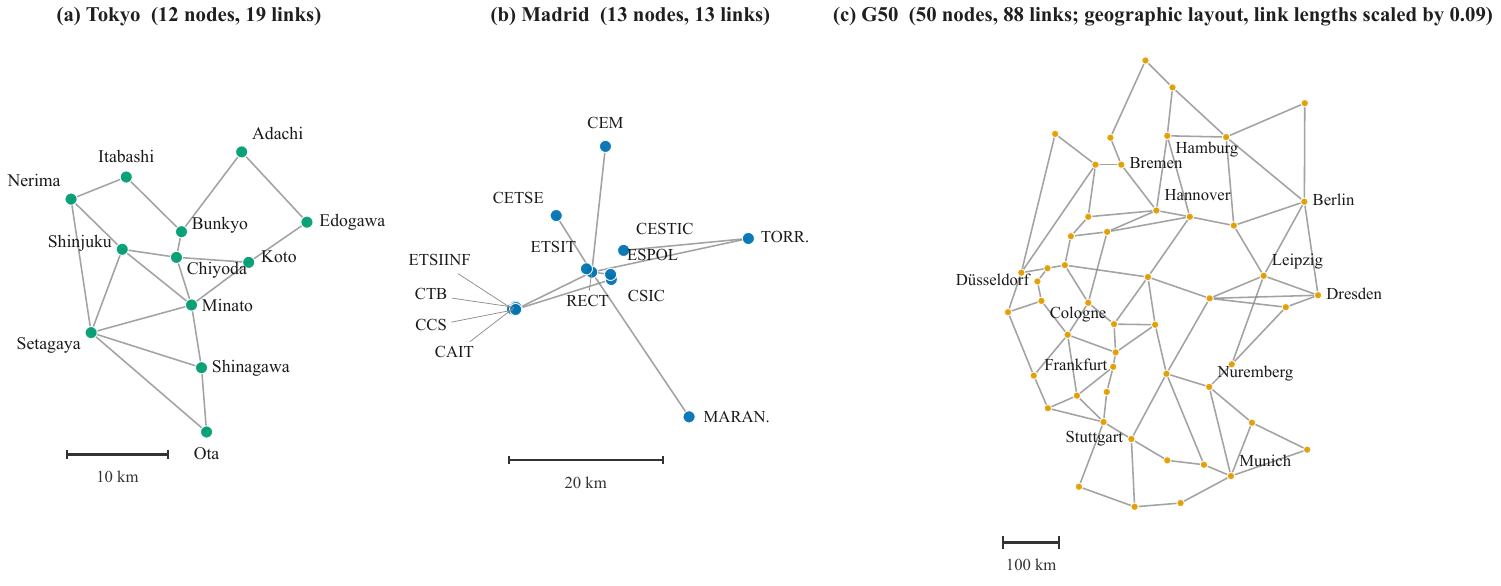}
    \caption{The three evaluation topologies: (a)~Tokyo, (b)~Madrid, (c)~G50. Topology parameters are listed in Table~\ref{tab:topologies}.}
    \label{fig:topologies}
\end{figure*}

\begin{table}[t]
\centering
\caption{Topology parameters for the three evaluation networks.}
\label{tab:topologies}
\begin{tabular}{lccccc}
\toprule
\textbf{Topology} & \textbf{$|V|$} & \textbf{$|A|$} & \textbf{Avg.\ $\ell$} & \textbf{Diam.} & \textbf{Default $B$} \\
 & & & (km) & (km) & \\
\midrule
Tokyo~\cite{Sasaki2011} & 12 & 38 & 8.2 & 30 & 50 \\
Madrid~\cite{Martin2024} & 13 & 26 & 8.7 & 61 & 50 \\
G50~\cite{Orlowski2010} & 50 & 176 & 9.2 & 84 & 50 \\
\bottomrule
\end{tabular}
\end{table}

Table~\ref{tab:topologies} summarizes the main parameters of the three evaluation topologies. Tokyo~\cite{Sasaki2011} and Madrid~\cite{madrid_2025_advancing, Martin2024} adopt the geographic distances of their real metropolitan fiber infrastructures. G50~\cite{Orlowski2010} is a 50-node backbone topology whose link lengths are scaled to metropolitan working distances comparable to the other two topologies~\cite{Alleaume2009}. Throughout this work, the network diameter denotes the largest shortest-path distance between any node pair, measured in fiber kilometers rather than in hops, and therefore corresponds to the longest end-to-end span that a hub placement may have to cover. Tokyo is a dense urban mesh (average degree 3.2, diameter 30\,km) in which all node pairs lie within MDI-QKD reach, serving as a baseline where hub placement alone can eliminate all trusted relays. Madrid is a sparse, near-tree metropolitan network (average degree 2.0, diameter 61\,km) derived from the MadQCI infrastructure, where limited path redundancy makes MDI-QKD coverage particularly sensitive to the deployment model. G50 provides metropolitan-scale link lengths comparable to those of Tokyo and Madrid but at a larger network diameter (84\,km) and 50 nodes, enabling evaluation of scalability; some node pairs are structurally unreachable by MDI-QKD regardless of hub placement, revealing topology-limited performance floors.

\subsection{Simulation Scenarios}
\label{sec:experiments}

Unless otherwise stated, all MILP instances use the baseline parameter values listed in Table~\ref{tab:params}; each scenario below varies one of these parameters while the others remain at their default. The parameter values reflect commercial MDI-QKD practice: $K_\text{max}=3$ matches the multiplexing capacity of deployed shared-relay systems~\cite{Berrevoets2022}; the hub setup cost $c_h$ and per-use cost $c_u$ are normalized so that installing a BSM unit dominates per-connection usage; and the default budget $B=50$ is chosen so that the budget constraint becomes binding at intermediate request loads rather than being trivially slack or globally infeasible.

\begin{table}[t]
\centering
\caption{Default MILP parameters used across all simulation scenarios.}
\label{tab:params}
\begin{tabular}{lll}
\toprule
\textbf{Parameter} & \textbf{Value} & \textbf{Description} \\
\midrule
$c_h$ & 2.0 & BSM hub setup cost \\
$c_u$ & 0.5 & Per-use measurement cost \\
$K_\text{max}$ & 3 & Hub capacity (requests/unit) \\
$\kappa_\text{min}$ & 10\,bps & Min.\ key-rate threshold \\
$B$ & 50 & Default budget \\
$d_\text{max}$ & 250\,km & MDI-QKD reach limit \\
$W$ & 20 & Quantum channel capacity per arc \\
$b_l$ & 0.5\,dB & Bypass loss per node \\
$\alpha, \beta, \gamma$ & 1000, 1, 0.001 & Objective weights \\
\bottomrule
\end{tabular}
\end{table}

We design four simulation scenarios by varying key parameters while keeping all other parameters at their baseline values:
\begin{itemize}
    \item \textbf{ExpA (Key-Rate Sweep):} We vary $\kappa_\text{min} \in \{10,\allowbreak 100,\allowbreak 200,\allowbreak 400,\allowbreak 700,\allowbreak 1000,\allowbreak 1300\}$\,bps with $|\mathcal{R}| = 20$ to study how stricter key-rate thresholds affect relay deployment feasibility across topologies of different diameter.

    \item \textbf{ExpB (Joint Load and Key-Rate Sweep):} We jointly vary $|\mathcal{R}| \in \{5, \ldots, 35\}$ and $\kappa_\text{min} \in \{10,\allowbreak 100,\allowbreak 200,\allowbreak 400,\allowbreak 700,\allowbreak 1000,\allowbreak 1300\}$\,bps to map the operating region of pure MDI-QKD deployment.

    \item \textbf{ExpC (Budget Sweep):} We vary the budget $B \in \{15,\allowbreak 20,\allowbreak \ldots,\allowbreak 50,\allowbreak 60,\allowbreak 75,\allowbreak 100\}$ at $|\mathcal{R}|=35$ and $\kappa_\text{min}=400$\,bps to identify the minimum investment required for a feasible pure MDI-QKD deployment.

    \item \textbf{ExpD (Network-Scale Sweep):} We uniformly scale the G50 link lengths so that the network diameter spans 84 to 252\,km, with $|\mathcal{R}| = 20$ and $\kappa_\text{min} \in \{10,\allowbreak 200,\allowbreak 400,\allowbreak 700,\allowbreak 1000\}$\,bps, to isolate the effect of geographic scale from that of topology structure.
\end{itemize}

Each scenario is executed for all three topologies, three candidate strategies (S1, S2, S3), and both deployment models (loss-compensated, loss-uncompensated), yielding $3 \times 3 \times 2 = 18$ configurations per parameter point.

All instances are solved with a per-instance time limit of 60\,s. Across the evaluated instances, the median solve time is below one second and most complete within a few seconds; the exceptions are large loss-uncompensated instances, whose broad candidate-path sets enlarge the model. The MILP is solved to proven optimality for the majority of feasible instances, confirming that the per-path hub-activation formulation is tractable at metropolitan scale despite the NP-hardness of the underlying problem.

\begin{figure*}[t]
    \centering
    \includegraphics[width=0.94\textwidth]{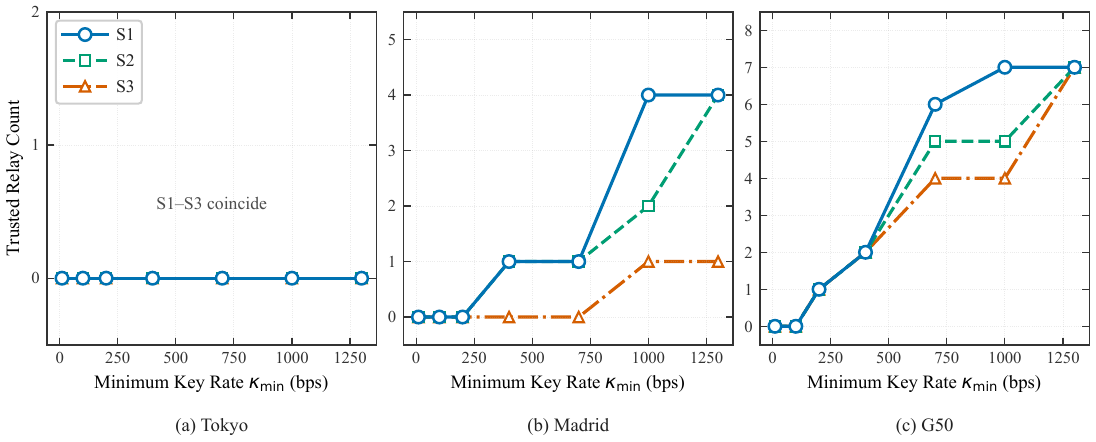}
\caption{Trusted-relay count vs.\ minimum key-rate threshold under three candidate strategies (loss-uncompensated, $|\mathcal{R}|=20$, $B=50$). The onset of nonzero TR shifts with network diameter.}
    \label{fig:threshold}
\end{figure*}
\section{Results and Analysis}
\label{sec:results}

\subsection{Impact of Key-Rate Threshold and Candidate Strategy}
\label{sec:threshold}

Fig.~\ref{fig:threshold} shows how trusted-relay usage grows with the minimum key-rate threshold $\kappa_\text{min}$ under three candidate strategies (S1--S3) for each topology ($|\mathcal{R}|=20$, $B=50$, loss-uncompensated). In Tokyo (panel~a), the compact topology (diameter 30\,km) keeps all node pairs within MDI-QKD reach: all strategies achieve TR\,=\,0 across the entire threshold range up to $\kappa_\text{min}=1300$\,bps, confirming that hub placement alone eliminates all trusted relays in dense metropolitan meshes. In Madrid (panel~b), the sparser structure and larger diameter (61\,km) cause trusted relays to appear at approximately 400\,bps, rising to TR\,=\,4 (S1, S2) and TR\,=\,1 (S3) at $\kappa_\text{min}=1300$. In G50 (panel~c), the largest diameter (84\,km) makes the model most sensitive to $\kappa_\text{min}$: TR begins increasing at $\kappa_\text{min} \approx 200$\,bps and reaches TR\,=\,7 under all three strategies at $\kappa_\text{min}=1300$. The onset of nonzero TR thus shifts monotonically with network diameter, from 1300+\,bps (Tokyo) through $\sim$400\,bps (Madrid) to $\sim$200\,bps (G50). Across all topologies, the ordering S1\,$\geq$\,S2\,$\geq$\,S3 is consistent, confirming that richer candidate placement reduces trusted-relay demand at every operating point.

\subsection{Impact of Network Scale}
\label{sec:scale}

\begin{figure}[t]
    \centering
    \includegraphics[width=0.84\columnwidth]{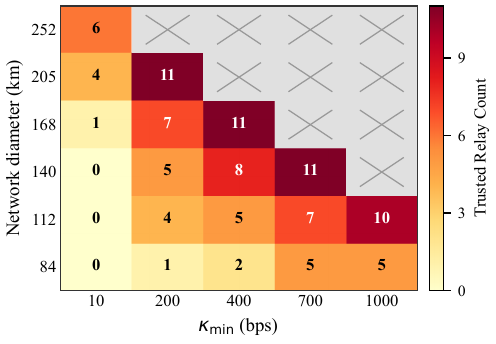}
    \caption{Trusted-relay count for G50 under uniform scaling of all link lengths (loss-uncompensated, S2, $|\mathcal{R}|=20$, $B=50$). Gray crosses ($\times$) mark infeasible configurations.}
    \label{fig:scale}
\end{figure}

The comparison across the three topologies suggests that network diameter is the primary driver of trusted-relay demand, yet the topologies also differ in connectivity structure. To isolate the effect of scale, ExpD uniformly scales all G50 link lengths while keeping the node set, connectivity, and request set fixed, so that only the geographic extent of the network changes. Fig.~\ref{fig:scale} shows the resulting trusted-relay count as the diameter grows from 84 to 252\,km. Trusted-relay usage increases monotonically along both axes, and the infeasibility frontier advances diagonally: at diameters up to roughly 110\,km, all evaluated key-rate targets remain feasible, whereas at larger scales the most demanding targets become progressively infeasible under the given budget. This controlled sweep reproduces, within a single topology, the diameter-governed onset observed across topologies in Fig.~\ref{fig:threshold}, confirming that geographic scale is the dominant driver of the feasibility boundary of pure MDI-QKD deployment. Under loss-compensated deployment the same sweep reaches infeasibility at smaller scales, since the loss-balancing window tightens relative to the growing link lengths, consistent with the behavior observed in Madrid.

\subsection{Impact of Deployment Budget}
\label{sec:budget}
Varying the budget (ExpC, $|\mathcal{R}|=35$, $\kappa_\text{min}=400$\,bps) reveals a sharp transition governed by a critical value $B^*$: below $B^*$, no budget-feasible pure MDI-QKD deployment exists, because the budget cannot fund the BSM units needed to serve all requests via two-link connections; at or above $B^*$, trusted-relay usage immediately saturates at a topology- and model-dependent floor. Tokyo and G50 become feasible at $B^*=45$ under all strategies, after which Tokyo saturates at TR\,=\,0 (its density admits a qualifying hub for every pair) while G50 saturates at TR\,=\,3 (a few node pairs exceed MDI-QKD range regardless of budget). Madrid exposes the cost of loss balancing most clearly: under loss-compensated (Comp.) deployment it requires $B^*=60$ (S1, S2) and saturates at TR\,=\,9 and 8, whereas the loss-uncompensated (Uncomp.) deployment lowers the threshold to $B^*=45$ and the floor to TR\,=\,2, and the richest candidate set (S3) drives Madrid to TR\,=\,3 under Comp.\ and TR\,=\,0 under Uncomp. This infeasible-to-saturation transition shows that pure MDI-QKD deployment requires a minimum infrastructure investment before any service is possible, and that the required investment grows with the loss-balancing constraint in sparse topologies. The choice of a default $B=50$ therefore sits just above $B^*$ for most configurations, keeping the budget constraint operative rather than vacuous.

\subsection{Impact of Deployment Model}
\label{sec:symmetry}

Fig.~\ref{fig:symmetry} maps trusted-relay usage across the ($|\mathcal{R}|$, $\kappa_\text{min}$) operating space under S2 candidate placement, with the key-rate axis extended to 1300\,bps. In Tokyo, both deployment models yield TR\,=\,0 across the entire operating region, confirming that in compact metropolitan meshes, hub placement alone is sufficient. In Madrid, the Comp.\ and Uncomp.\ deployments diverge: under Comp., TR rises with load and becomes infeasible for $|\mathcal{R}|=35$ at all $\kappa_\text{min}$, whereas Uncomp. remains feasible with substantially lower TR throughout, confirming that the loss-balancing constraint is the dominant cost driver in sparse topologies. In G50, both deployment models exhibit similar TR (the dense connectivity makes loss balancing nearly free), and the feasibility boundary is instead driven by high $\kappa_\text{min}$ combined with high load, reflecting the larger network diameter. Across all topologies, relaxing the loss-balancing constraint either reduces trusted-relay usage (Madrid) or has no effect (Tokyo, G50), making it a cost-free or cost-effective planning lever.

\begin{figure*}[t]
    \centering
    \includegraphics[width=0.87\textwidth]{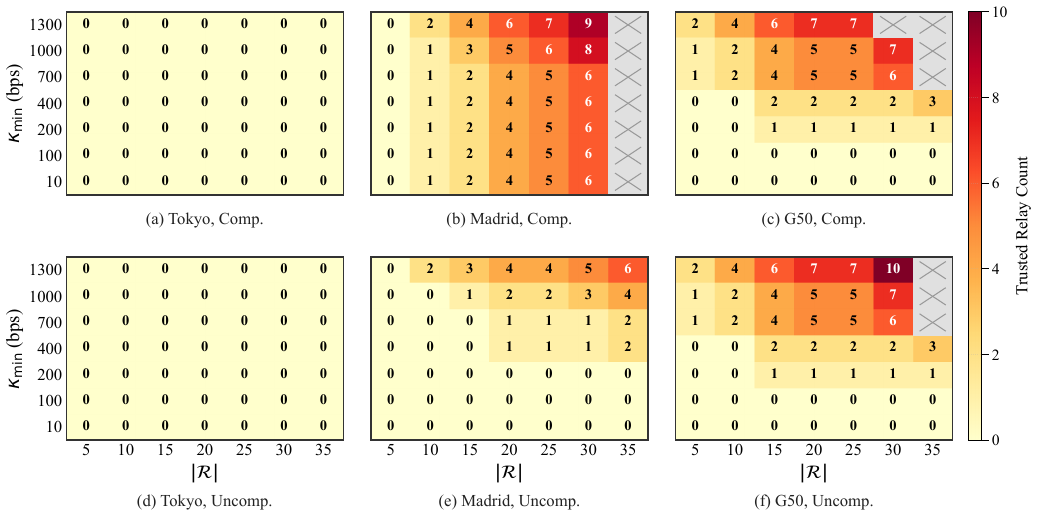}
    \caption{Trusted-relay count under S2 candidate placement ($B=50$). Gray crosses ($\times$) mark infeasible configurations. Top row: loss-compensated; bottom row: loss-uncompensated. The key-rate axis extends to 1300\,bps to capture the full operating region.}
    \label{fig:symmetry}
\end{figure*}

The mechanism is structural. In Madrid, hub candidates satisfying the loss-balancing constraint are scarce: under Comp.\ S2, TR reaches 4 at $|\mathcal{R}|=20$ and $\kappa_\text{min}=10$, whereas Uncomp.\ achieves TR\,=\,0 under the same conditions, eliminating all trusted relays at no additional infrastructure cost. In Tokyo and G50, denser connectivity ensures that near-midpoint hub candidates are almost always available, so loss balancing is effectively free.

\section{Discussion}
\label{sec:discussion}

The results suggest several practical implications for MDI-QKD network planning. First, in compact metropolitan networks such as Tokyo (diameter 30\,km), hub placement alone eliminates all trusted relays regardless of the key-rate threshold or deployment model, confirming that pure MDI-QKD deployment is straightforward at this scale. Second, in sparse topologies such as Madrid, relaxing the loss-balancing constraint eliminates trusted relays that would otherwise be required under loss-compensated deployment, making it the most cost-effective planning lever because it requires no additional infrastructure. In denser topologies, loss balancing is effectively free. Third, increasing candidate density (S1$\to$S3) consistently reduces TR across all topologies, though S3 remains a theoretical reference given its uncaptured new-fiber cost. Finally, the key-rate threshold at which trusted relays first appear shifts monotonically with network diameter, from beyond 1300\,bps in Tokyo to approximately 200\,bps in G50, and the controlled scale sweep of Fig.~\ref{fig:scale} confirms this dependence within a single topology, indicating that QoS targets must be matched to network scale.
Several modeling simplifications should be noted. The key-rate model uses an experimental lookup table~\cite{Berrevoets2022} with asymptotic rates, and tighter security analyses incorporating finite-key effects~\cite{Curty2014} would reduce the achievable rates. Since the key rate enters the model solely through the viability condition $\kappa(d_\text{eff}) \geq \kappa_\text{min}$, a uniform rate reduction by a factor $\eta$ is exactly equivalent to raising the threshold to $\kappa_\text{min}/\eta$; the sensitivity of our results to finite-key corrections can therefore be read directly from the key-rate sweeps in Figs.~\ref{fig:threshold} and~\ref{fig:symmetry}. At the default 10\,bps threshold, the optimal trusted-relay count is unchanged under a tenfold rate reduction across all evaluated topologies and both deployment models. The effective-distance metric $d_\text{eff} = (2/\alpha_\text{dB})\max(L_A, L_B)$ is likewise conservative under loss-uncompensated deployment, since evaluating the key rate at the worse of the two channels over-estimates the loss whenever the two user-to-hub paths differ. In fact, the two deployment models bracket any per-link optimized operating point: the loss-uncompensated rate is achievable today with passive attenuation, while evaluating the key rate at the total loss of the two channels, i.e., $d_\text{eff} = (L_A + L_B)/\alpha_\text{dB}$, upper-bounds what asymmetric MDI-QKD protocols~\cite{Wang2019, Liu2019} could achieve. Re-solving the ExpA and ExpB instances under this upper bound leaves the optimal trusted-relay count unchanged across the evaluated instances with $\kappa_\text{min} \leq 100$\,bps, so the reported results are unaffected by protocol-level optimization in this regime. At more demanding thresholds in the sparse and large-diameter topologies, the bound lowers TR by at most three and marginally extends the feasible region, quantifying the maximum residual benefit of per-link protocol optimization, which we leave for future work.

\section{Conclusion}
\label{sec:conclusion}

In this work, we formalized the MDI-QKD Hub Deployment (MHD) problem, proved its NP-hardness, and proposed a MILP formulation that jointly optimizes BSM-hub placement, two-link user-to-hub routing, and trusted-relay assignment under budget and capacity constraints. Evaluation across three metropolitan topologies (12 to 50 nodes) at realistic geographic distances reveals that topology structure governs trusted-relay demand: compact meshes eliminate all trusted relays through hub placement alone, sparse topologies benefit from relaxing the loss-balancing constraint, and the key-rate threshold at which trusted relays first appear shifts monotonically with network diameter. The formulation delineates the feasibility boundary of pure MDI-QKD deployment, providing actionable guidance for network planners. Future work will extend the model to dynamic demand scenarios and larger networks via scalable heuristics.

\section*{Acknowledgment}

This work was supported in part by funding from the Innovation for Defence Excellence and Security (IDEaS) program from the Department of National Defence (DND); the EU Horizon Europe project ``Quantum Secure Networks Partnership'' (QSNP), grant 101114043; and the ``Hub Nacional de Excelencia en Comunicaciones Cuánticas'' project, funded by Ministerio para la Transformación Digital y de la Función Pública and by the Recovery, Transformation and Resilience Plan -- Funded by the European Union -- NextGenerationEU (PRTR-C16.R1).

\bibliographystyle{IEEEtran}
\IEEEtriggeratref{20}
\bibliography{references}

\begin{thebibliography}{10}
\providecommand{\url}[1]{#1}
\csname url@samestyle\endcsname
\providecommand{\newblock}{\relax}
\providecommand{\bibinfo}[2]{#2}
\providecommand{\BIBentrySTDinterwordspacing}{\spaceskip=0pt\relax}
\providecommand{\BIBentryALTinterwordstretchfactor}{4}
\providecommand{\BIBentryALTinterwordspacing}{\spaceskip=\fontdimen2\font plus
\BIBentryALTinterwordstretchfactor\fontdimen3\font minus \fontdimen4\font\relax}
\providecommand{\BIBforeignlanguage}[2]{{%
\expandafter\ifx\csname l@#1\endcsname\relax
\typeout{** WARNING: IEEEtran.bst: No hyphenation pattern has been}%
\typeout{** loaded for the language `#1'. Using the pattern for}%
\typeout{** the default language instead.}%
\else
\language=\csname l@#1\endcsname
\fi
#2}}
\providecommand{\BIBdecl}{\relax}
\BIBdecl

\bibitem{Bennett1984}
C.~H. Bennett and G.~Brassard, ``Quantum cryptography: Public key distribution and coin tossing,'' in \emph{Proc. IEEE Int. Conf. Comput. Syst. Signal Process.}, Bangalore, India, 1984, pp. 175--179.

\bibitem{Lo2014}
H.-K. Lo, M.~Curty, and K.~Tamaki, ``Secure quantum key distribution,'' \emph{Nature Photon.}, vol.~8, no.~8, pp. 595--604, 2014.

\bibitem{Pirandola2017}
S.~Pirandola, R.~Laurenza, C.~Ottaviani, and L.~Banchi, ``Fundamental limits of repeaterless quantum communications,'' \emph{Nature Commun.}, vol.~8, no. 15043, 2017.

\bibitem{qkd_a_networking}
M.~Mehic, M.~Niemiec, S.~Rass, J.~Ma, M.~Peev, A.~Aguado, V.~Martin, S.~Schauer, A.~Poppe, C.~Pacher, and M.~Voznak, ``Quantum key distribution: A networking perspective,'' \emph{ACM Comput. Surv.}, vol.~53, no.~5, pp. 96:1--96:41, 2020.

\bibitem{Zhao2018}
Y.~Zhao, Y.~Cao, W.~Wang, H.~Wang, X.~Yu, J.~Zhang, M.~Tornatore, Y.~Wu, and B.~Mukherjee, ``Resource allocation in optical networks secured by quantum key distribution,'' \emph{IEEE Commun. Mag.}, vol.~56, no.~8, pp. 130--137, 2018.

\bibitem{Zhang2024RCKTA}
Q.~Zhang, O.~Ayoub, A.~Gatto, J.~Wu, F.~Musumeci, and M.~Tornatore, ``Routing, channel, key-rate, and time-slot assignment for {QKD} in optical networks,'' \emph{IEEE Trans. Netw. Serv. Manag.}, vol.~21, no.~1, pp. 148--160, 2024.

\bibitem{xiong2025power}
J.~Xiong, Q.~Zhang, Y.~Pi{\'e}tri, R.~Yehia, R.~Boutaba, F.~Musumeci, and M.~Tornatore, ``Power consumption analysis of {QKD} networks under different protocols and detector configurations,'' in \emph{Proc. Eur. Conf. Opt. Commun. (ECOC)}, 2025.

\bibitem{web_euroqci}
``{The European Quantum Communication Infrastructure (EuroQCI) Initiative},'' \url{https://digital-strategy.ec.europa.eu/en/policies/european-quantum-communication-infrastructure-euroqci}, 2024.

\bibitem{euroqci_QUID_Italia}
{QUID Consortium}, ``Quid website,'' \url{https://quid-euroqci-italy.eu/}, 2025, accessed: 2025-11-30.

\bibitem{euroqci_spain}
{EuroQCI-Spain Consortium}, ``Euroqci-spain website,'' \url{https://euroqci-spain.eu/}, 2025, accessed: 2025-11-30.

\bibitem{Lo2012}
H.-K. Lo, M.~Curty, and B.~Qi, ``Measurement-device-independent quantum key distribution,'' \emph{Phys. Rev. Lett.}, vol. 108, no. 130503, 2012.

\bibitem{Lydersen2010}
L.~Lydersen, C.~Wiechers, C.~Wittmann, D.~Elser, J.~Skaar, and V.~Makarov, ``Hacking commercial quantum cryptography systems by tailored bright illumination,'' \emph{Nature Photon.}, vol.~4, no.~10, pp. 686--689, 2010.

\bibitem{Braunstein2012}
S.~L. Braunstein and S.~Pirandola, ``Side-channel-free quantum key distribution,'' \emph{Phys. Rev. Lett.}, vol. 108, no. 130502, 2012.

\bibitem{Tang2014}
Y.-L. Tang \emph{et~al.}, ``Measurement-device-independent quantum key distribution over 200\,km,'' \emph{Phys. Rev. Lett.}, vol. 113, no. 190501, 2014.

\bibitem{Tang2016}
Y.-L. Tang, H.-L. Yin, Q.~Zhao, H.~Liu, X.-X. Sun, M.-Q. Huang, W.-J. Zhang, S.-J. Chen, L.~Zhang, L.-X. You, Z.~Wang, Y.~Liu, C.-Y. Lu, X.~Jiang, X.~Ma, Q.~Zhang, T.-Y. Chen, and J.-W. Pan, ``Measurement-device-independent quantum key distribution over untrustful metropolitan network,'' \emph{Phys. Rev. X}, vol.~6, no. 011024, 2016.

\bibitem{Berrevoets2022}
R.~C. Berrevoets \emph{et~al.}, ``Deployed measurement-device independent quantum key distribution and {Bell}-state measurements coexisting with standard internet data and networking equipment,'' \emph{Commun. Phys.}, vol.~5, no. 186, 2022.

\bibitem{qbird_web}
{Q*Bird}, ``{Falqon MDI-QKD Platform},'' \url{https://q-bird.com/mdi-qkd-falqon-series/}, 2025, accessed: 2026-04-24.

\bibitem{Alumur2008}
S.~Alumur and B.~Y. Kara, ``Network hub location problems: The state of the art,'' \emph{Eur. J. Oper. Res.}, vol. 190, no.~1, pp. 1--21, 2008.

\bibitem{OKelly1987}
M.~E. O'Kelly, ``A quadratic integer program for the location of interacting hub facilities,'' \emph{Eur. J. Oper. Res.}, vol.~32, no.~3, pp. 393--404, 1987.

\bibitem{hong1987measurement}
C.-K. Hong, Z.-Y. Ou, and L.~Mandel, ``Measurement of subpicosecond time intervals between two photons by interference,'' \emph{Physical review letters}, vol.~59, no.~18, pp. 2044--2046, 1987.

\bibitem{Wang2019}
W.~Wang, F.~Xu, and H.-K. Lo, ``Asymmetric protocols for scalable high-rate measurement-device-independent quantum key distribution networks,'' \emph{Phys. Rev. X}, vol.~9, no. 041012, 2019.

\bibitem{Liu2019}
H.~Liu \emph{et~al.}, ``Experimental demonstration of high-rate measurement-device-independent quantum key distribution over asymmetric channels,'' \emph{Phys. Rev. Lett.}, vol. 122, no. 160501, 2019.

\bibitem{Cao2021}
Y.~Cao, Y.~Zhao, J.~Li, R.~Lin, J.~Zhang, and J.~Chen, ``Hybrid trusted/untrusted relay-based quantum key distribution over optical backbone networks,'' \emph{IEEE J. Sel. Areas Commun.}, vol.~39, no.~9, pp. 2701--2718, 2021.

\bibitem{Jia2023}
J.~Jia, B.~Dong, L.~Kang, H.~Xie, and B.~Guo, ``Cost-optimization-based quantum key distribution over quantum key pool optical networks,'' \emph{Entropy}, vol.~25, no.~4, p. 661, 2023.

\bibitem{He2025}
J.~He, Y.~Zhou, S.~Xie, and C.~Wang, ``A cost-optimal quantum key distribution based on hybrid trusted relays,'' in \emph{Proc. BDCTA}, 2025, pp. 130--138.

\bibitem{Rabbie2022}
J.~Rabbie, K.~Chakraborty, G.~Avis, and S.~Wehner, ``Designing quantum networks using preexisting infrastructure,'' \emph{npj Quantum Inf.}, vol.~8, no.~5, 2022.

\bibitem{Pouryousef2024}
S.~Pouryousef, H.~Shapourian, A.~Shabani, R.~Kompella, and D.~Towsley, ``Resource placement for rate and fidelity maximization in quantum networks,'' \emph{IEEE Trans. Quantum Eng.}, vol.~5, pp. 1--16, 2024.

\bibitem{boston}
C.~Elliott, A.~Colvin, D.~Pearson, O.~Pikalo, J.~Schlafer, and H.~Yeh, ``Current status of the {DARPA} quantum network,'' in \emph{Quantum Information and Computation III (SPIE)}, vol. 5815, 2005, pp. 138--149.

\bibitem{Peev2009}
M.~Peev \emph{et~al.}, ``The {SECOQC} quantum key distribution network in {Vienna},'' \emph{New J. Phys.}, vol.~11, no. 075001, 2009.

\bibitem{Sasaki2011}
M.~Sasaki \emph{et~al.}, ``Field test of quantum key distribution in the {Tokyo} {QKD} network,'' \emph{Opt. Express}, vol.~19, no.~11, pp. 10\,387--10\,409, 2011.

\bibitem{Chen2021}
Y.-A. Chen \emph{et~al.}, ``An integrated space-to-ground quantum communication network over 4,600 kilometres,'' \emph{Nature}, vol. 589, pp. 214--219, 2021.

\bibitem{madrid_2025_advancing}
A.~Sebasti{\'a}n-Lombra{\~n}a, L.~Ortiz, J.~P. Brito, J.~Faba, R.~B. M{\'e}ndez, J.~S. {De Buruaga}, R.~J. Vicente, J.~Setien, J.~J. Romero, C.~Escribano, P.~Salas, J.~L. Bejarano, and V.~Mart{\'i}n, ``Advancing the future of quantum communication networks: the new {MadQCI},'' in \emph{Proc. 25th Int. Conf. Transparent Opt. Netw. (ICTON)}, 2025, pp. 1--5.

\bibitem{Zhang2023ICQKD}
Q.~Zhang, O.~Ayoub, J.~Wu, X.~Lin, and M.~Tornatore, ``{IC-QKD}: An information-centric quantum key distribution network,'' \emph{IEEE Commun. Mag.}, vol.~61, no.~12, pp. 148--154, 2023.

\bibitem{Cirigliano2024}
L.~Cirigliano, V.~Brosco, C.~Castellano, C.~Conti, and L.~Pilozzi, ``Optimal quantum key distribution networks: Capacitance versus security,'' \emph{npj Quantum Inf.}, vol.~10, no.~44, 2024.

\bibitem{Amer2020}
O.~Amer, W.~O. Krawec, and B.~Wang, ``Efficient routing for quantum key distribution networks,'' in \emph{Proc. IEEE Int. Conf. Quantum Comput. Eng. (QCE)}, 2020, pp. 137--147.

\bibitem{Krawec2024}
W.~O. Krawec, B.~Wang, and R.~Brown, ``Finite key security of simplified trusted node networks,'' in \emph{Proc. IEEE Int. Conf. Quantum Comput. Eng. (QCE)}, 2024, pp. 1777--1787.

\bibitem{Marik2025}
A.~K. Marik, B.~Palit, and S.~Behera, ``Trusted repeater placement in {QKD}-enabled optical networks,'' \emph{arXiv preprint arXiv:2509.10338}, 2025.

\bibitem{Avis2024}
G.~Avis, R.~Knegjens, A.~S. S{\o}rensen, and S.~Wehner, ``Asymmetric node placement in fiber-based quantum networks,'' \emph{Phys. Rev. A}, vol. 109, no. 052627, 2024.

\bibitem{Karp1972}
R.~M. Karp, ``Reducibility among combinatorial problems,'' in \emph{Complexity of Computer Computations}, R.~E. Miller, J.~W. Thatcher, and J.~D. Bohlinger, Eds.\hskip 1em plus 0.5em minus 0.4em\relax Plenum Press, 1972, pp. 85--103.

\bibitem{Martin2024}
V.~Martin \emph{et~al.}, ``{MadQCI}: A heterogeneous and scalable {SDN} {QKD} network deployed in production facilities,'' \emph{npj Quantum Inf.}, vol.~10, no.~80, 2024.

\bibitem{Orlowski2010}
S.~Orlowski, R.~Wess{\"a}ly, M.~Pi{\'o}ro, and A.~Tomaszewski, ``{SNDlib} 1.0---survivable network design library,'' \emph{Networks}, vol.~55, no.~3, pp. 276--286, 2010.

\bibitem{Alleaume2009}
R.~All{\'e}aume, F.~Roueff, E.~Diamanti, and N.~L{\"u}tkenhaus, ``Topological optimization of quantum key distribution networks,'' \emph{New J. Phys.}, vol.~11, no. 075002, 2009.

\bibitem{Curty2014}
M.~Curty, F.~Xu, W.~Cui, C.~C.~W. Lim, K.~Tamaki, and H.-K. Lo, ``Finite-key analysis for measurement-device-independent quantum key distribution,'' \emph{Nat. Commun.}, vol.~5, no. 3732, 2014.

\end{thebibliography}

\end{document}